\documentclass[letterpaper,12pt,margin=0.5in]{article}
\usepackage[letterpaper, top=1.6in, bottom=1.6in,left=1in,right=1in]{geometry}
\usepackage[utf8]{inputenc}
\usepackage{amsmath,amsthm,amssymb,array,xcolor,multicol,verbatim,mathpazo}
\usepackage{hyperref}
\usepackage{tikz}
\usepackage{subcaption}
\usepackage{cleveref}
\usepackage{amssymb}
\usepackage{pgfplots}
\usetikzlibrary{positioning,angles}
\usetikzlibrary{arrows}\usetikzlibrary{shapes.multipart}

\usepackage{forest}
\usepackage[show]{ed}
\usepackage[
style=authoryear,
]{biblatex}
\usetikzlibrary{patterns,positioning,shapes.geometric}
\usepackage{graphicx}
\usepackage{enumitem}
\usepackage{algorithm}
\usepackage[noend]{algpseudocode}

\usetikzlibrary{calc}

\tikzset{axis line style/.style={thin, gray, -stealth}}
\usepackage{adjustbox}
\usepackage{array}

\usepackage{macros}

\newtheorem{thm}{Theorem}

\newtheorem{prop}{Proposition}
\newtheorem*{fact}{Observation}

\newtheoremstyle{named}{}{}{\itshape}{}{\bfseries}{.}{.5em}{\thmnote{#3}}

\theoremstyle{definition}
\newtheorem{ex}{Example}

\newtheorem{defn}{Definition}

\makeatother
\newcounter{parentnumber}
\usepackage{amsmath}
\usepackage{newtxtext,newtxmath}
\DeclareSymbolFont{CMletters}{OML}{cmm}{m}{it}
\DeclareMathSymbol{v}{\mathord}{CMletters}{`v}
\DeclareMathOperator*{\argmax}{arg\,max}

\newtheorem{cor}{Corollary}
\title{Purification and Perturbations of Communication and Repeated Games}
\date{Updated: \today
\\\href{https://drive.google.com/file/d/1UZBn1w8nszj_posHj9AtPISpueDrl8WP/view?usp=sharing}{Click here} to access current version
}
\author{Alistair Barton\thanks{Cambridge University, Email: \href{mailto:ab3315@cam.ac.uk}{ab3315@cam.ac.uk}\\
I thank Joyee Deb, Dilip Abreu, Ariel Rubinstein, Kirtivardhan Singh, and Diego Cussen for productive feedback on this project.
}}
\begin{document}
\maketitle
\begin{abstract}
I prove that it is irrational for agents with even slightly private preferences to condition their strategy on private information that is payoff-irrelevant to them, contrary to powerful techniques for analyzing communication and repeated games. In repeated games with public+private monitoring, this means all pure equilibria are perfect public equilibria, and non-trivial belief free equilibria do not exist. In a wide class of communication games (up to allowing receiver commitment), this means persuasion is impossible with state-independent preferences. Nevertheless, these analytic techniques may be made compatible with private preferences through perturbation approaches: considering either payoff-relevance of the private information or correlation between parties' information. An example of the latter occurs by introducing `atonement' to repeated games equilibria.

\end{abstract}

Game theory typically assumes the preferences are public; however, in reality preferences are, to some extent, private. Even in the relatively simple setting of monetary lotteries, it is not always predictable how an agent will compare two lotteries, as this valuation varies heterogeneously across agents depending on their risk preferences (and cognitive noise).

Neglecting such concerns is justified by Harsanyi's purification theorem (\cite{H73}), which states that slight privacy does not significantly disturb the set of Nash equilibria in generic games. Such privacy will give almost every individual a strict best response, thus affecting their best response in mixed equilibria, aggregating across preferences the distribution of actions will approximate the original equilibrium.

In this paper we focus on a commonly studied and popular (but still measure-0) set of games that are an exception to Harsanyi's purification theorem: even slight (private) heterogeneity in agent preferences significantly disturbs the set of equilibria. This includes both communication games with state-independent preferences (e.g. \cite{CH10,LR20,CD23,HL25}),\footnote{Related work by \cite{DK21} will be discussed in the literature review.} and repeated games --- with severe implications on the structure of equilibria when monitoring is imperfect.


We study a simple mechanism design model with state-independent preferences which nests communication and certain equilibria of repeated games. State-independent preferences significantly aid the tractable analysis of equilibria, but directly leads to fragility. Our main result can be informally stated as a principle that holds under a weak conditions on agents' private preferences called prevalence:
\begin{quote}
    \textbf{Independence of Irrelevant Information (III) (\cite{M91}):} Agents do not condition their strategy on information unless (i) it is payoff-relevant to \textit{them}, or (ii) it is informative of others' strategies. 
\end{quote}
While this is an intuitive property for equilibria, powerful and tractable results can be obtained by violating it --- we refer to such violations as `abusing indifference'. 

Note that the information referred to in III may be payoff-relevant to other agents, as in the following example demonstrating the power and fragility of abusing indifference:
\begin{ex}[Money-Burning as Persuasion]\label{ex:mb}
    A seller recommends a product from their inventory $\C:=\{c_0,\dots, c_N\}$ to a consumer. The seller privately observes the state $s\in S$ generated from prior $\mu\in\Delta S$ with full support, before making a recommendation $a=(a_c,a_m)$. The recommendation takes the form of a suggested product $a_c\in \C$ accompanied with an amount of money $a_m\ge 0$ that is publicly `burnt'. Based on the message $(a_c,a_m)$, based on the recommendation the consumer either chooses a product $a_c$ or walks away (represented by the `product' $c_0$).

    The seller's preference is state-independent, only considering the profit from the interaction, given by 
    $$
    u(c,a_c,a_m)=\pi_C(c)-a_m
    $$
    where $\pi_C(c)$ is the profit from selling product $c$, assumed to be positive: $\pi_C(c)\ge\pi_C(c_0)=0$. The buyer's \textit{ex post} utility, denoted $u_R(c\rvert s)$, depends on both the state and the product consumed. We assume that the buyer's best response is not constant over states, and normalize their \textit{ex ante} utility from babbling, $\max_c \{\int u_R(c\rvert s)\,d\mu(s)\}=0$, and their first-best utility, $\int \max_c\{u_R(c\rvert s)\}\,d\mu(s)=1$. Before the seller makes any recommendation, the consumer commits to a strategy $\hat{c}$ determining which product to purchase given each recommendation $a=(a_c,a_m)$.

    For the seller's strategy to be optimal, equilibrium recommendations must all result in the same total profit $u$. This occurs if the buyer commits to only obey a suggestion $a_c$ if $a_m=\pi(a_c)$, and otherwise purchases $c_0$. In this case every recommendation strategy for the seller is optimal. In particular it is optimal for the seller to fully reveal the state to the buyer, allowing them their first best utility $1$; however, it is also a best response for the seller to send a constant message. By interpolating between the two, the set of equilibrium buyer utilities is $[0,1]$.\footnote{The buyer cannot do worse than $0$ in any equilibrium, otherwise they would commit to the constant strategy that is their \textit{ex ante} best response. The result stated here can be obtained without commitment power. In Corollary \ref{cor:m.b} we return to this example, stating a stronger version of this result that does require commitment power.}

    Now suppose the seller's profit $\pi_C$ from each sale is private, distributed according to a density on $\R^\C$ (ie. the buyer does not know the mark-up on different products). For a fixed buyer strategy, the seller will then (w.p. 1) have a strictly preferred recommendation, independent of the state, which they will always recommend. Thus no information can be communicated in equilibrium and the buyer can do no better than their babbling utility $0$.
\end{ex}

In a world with private preferences, behaviour violating III is irrational. However, the goal of this paper is not to dismiss the powerful techniques that violate III, but rather present re-interpretations of these techniques that are consistent with private preferences. As a brief summary this paper argues:
\begin{enumerate}
    \item (Abusing Indifference) Behaviour violating III is not compatible with private preferences (assuming these preferences preserve payoff-irrelevance of the information).
    \item (Using Indifference) Equilibrium constructed by violating III may be made compatible with private preferences by perturbing either the correlation structure or payoff-relevance of information in a compatible manner. Beyond grounding this equilibrium behaviour, such augmentations can improve our understanding:
    \begin{enumerate}
        \item They help uncover forces necessary for supporting desired equilibrium,
        \item They extend analysis from Abusing Indifference beyond the measure-0 set of unperturbed models, to a large class of models where III cannot be directly applied (e.g. information is significantly payoff-relevant).
    \end{enumerate}
\end{enumerate}
Section \ref{sec:resolutions} contains examples and further discussion of how to interpret these equilibria.

In Section \ref{sec:app.rg} we study the implications of III in repeated games with imperfect monitoring (specifically monitoring that is a combination of a public and a private, conditionally independent, signal). We explore three types of perfect bayesian equilibria: Perfect Public Equilibria (PPE), Belief-Free Equilibria (BFE), and Belief-Based Equilibria (BBE). III implies that, for prevalent private utilities:
\begin{figure}\begin{center}
\begin{tikzpicture}[scale=1]
    \draw[color=black] (210:1.5) circle (2) node[below left,align=center]{Private\\ Pref.}; 
    \draw[color=black] (330:1.5) circle (2) node[below right,align=center]{Pure\\ Strategy};
    \draw[color=black] (90:1.5) circle (2) node[above, align=center]{Noisy Mon.};
    \node at (0,0) {PPE};
    \node at (1.04,0.6) {BFE+};
    \node at (-1.04,0.6) {BBE};
\end{tikzpicture}
\end{center}
    \caption{Our main result for repeated games shows how the introduction of private preferences interacts with noisy (public+private) monitoring to affect pure and non-pure equilibria (N.B. pure equilibria include public randomization), as summarized by this Venn Diagram. PPE are public perfect equilibria, BBE are belief-based equilibria, and BFE+ includes belief-free equilibria.}
    \label{fig:venn}
\end{figure}
\begin{itemize}
    \item all pure perfect bayesian equilibria are PPE (where pure means agents do not condition their action on their private preference type),
    \item non-trivial belief-free equilibria do not exist,
    \item any outcome that occurs w.p. 1 in a game with private monitoring must be myopically optimal.
\end{itemize}
The first two results are illustrated in Figure \ref{fig:venn}. 

PPE can be inefficient, some criticisms are that they force agents to neglect private information, impede efficient renegotiation proof equilibria, and lead to inefficient off-path play (as arise in cases of miscoordination, trembling-hand play, or large uncertainty about private preferences). This is because PPE preclude atonement: conditional on the same realization of the public signal, an agent that is guilty of deviating cannot act any different from an innocent agent. A particular application of this is the impossibility of symmetric renegotiation proof PPE when monitoring is anonymous. 

We discuss how Belief-Based Equilibria can be more efficient than PPE by introducing `randomization', and suggest that they can be tractably understood in settings with deterministic (imperfect) monitoring.
\subsubsection*{Related Literature}
Different instances of private sunspot arguments have appeared before. Our main contribution is to formally generalize, extend, and unify these critiques, allowing us to more broadly study the domain of models to which this criticism does/does not apply:

The principle of III was first described by \cite{M91} as an intuitive equilibrium criterion in repeated games with only private monitoring (not connected to purification), applied to obtain a result that pure equilibria are stage Nash. \cite{BMM09} (implicitly) apply the III principle recursively to sequential-move games with perfect information and finite memory to demonstrate that all purifiable equilibria are markov perfect equilibria. Similarly, \cite{DK21} apply this result to the cheap talk model of \cite{CH10} to demonstrate that only babbling equilibria of this model are compatible with private preferences. The main result of this paper more generally connects the III principle to private preferences, generalizing these previous results.

The novel applications in this work (to my knowledge) are to repeated games with public monitoring (as well as public+private monitoring, and perfect monitoring), and to more general communication games including mediation or receiver commitment. By outlining the full scope of this critique, we also demonstrate the necessary relaxations to obtain purifiable equilibria.

Another related paper is \cite{BMM08}, which studies purification in repeated games with perfect monitoring and 1-period memory. They do not use a private sunspot argument, and indeed they find that non-trivial equilibria can be compatible with private preferences. However, equilibrium behaviour is sensitive to the distribution of preferences (even in a neighbourhood of a public preference). This describes a failure of Harsanyi purification distinct from the failure we observe --- being about the sensitivity rather than the non-existence of non-trivial equilibria with private preferences.\footnote{As an analogy, imagine preferences are loosely represented by a non-injective matrix $A_0:\mathbb{R}^n\rightarrow \mathbb{R}^m$ and non-zero vector $b_0\in \text{im}(A_0)$, and non-trivial equilibria are represented by the (continuum) solutions $x^\ast(A_0,b_0)$ of $A_0x=b_0$. The \cite{BMM08} result can be compared to the case where $m=n$, so that generic perturbations to $A_0$ reduce the solution set to a unique $x^\ast$ which depends on the direction of the perturbation (ie. $(A,b)\mapsto x^\ast(A,b)$ is not lower-hemicontinuous at $(A_0,b_0)$). The private sunspot principle is comparable to the case $m>n$ where generic perturbations to $(A_0,b_0)$ result in \textbf{no} non-trivial equiliria.}
\section{Base Model}
We work with a maximally simple single-agent mechanism design problem where our question is which strategies the designer can incentivize in the agent. Richer environments (including both communication and repeated games) can be constructed from this model by imposing additional equilibrium constraints.

\begin{figure}
    \centering
    \begin{tikzpicture}[every node/.style={circle, text = black,minimum size=8.5mm},
    dots/.style={circle, text = black,minimum size=8.5mm}]
    \def\Tick{0.2}
    \def\displace{0.6}
        \draw[->,ultra thick] (0,0) -- (12,0);
        \draw (0,-\Tick) -- (0,\Tick);
        \node at (0,\displace) {Designer chooses};
        \node at (0,-\displace) {$\hat{c}:\A \rightarrow \Delta \C $};
        \draw (4,-\Tick) -- (4,\Tick);
        \node at (4,\displace) {Agent chooses};
        \node at (4,-\displace) {$\hat{a}:S\times\Omega\rightarrow \Delta \A $};
        \draw (8,-\Tick) -- (8,\Tick);
        \node at (8,\displace) {Nature chooses};
        \node[align=center] at (8,-\displace) {$\omega\sim\mathbb{P}$\\ $s$};
        \draw (12,-\Tick) -- (12,\Tick);
        \node at (12,\displace) {Nature chooses};
        \node[align=center] at (12,-0.8) {$a\sim \hat{a}(s\rvert\omega)$\\$c\sim \hat{c}(a)$};
    \end{tikzpicture}
    \caption{The timing of the game, the designer first commits to a mechanism $\hat{c}:\A \rightarrow \Delta\C $, the agent chooses a contingent strategy $\hat{a}:S\times\Omega\rightarrow \Delta \A $. The payoff for players in each state $s\in S$ is then determined by the realization of $\omega\sim\mathbb{P}$, $s$ $m\sim \hat{a}(s)$ and $c\sim\hat{c}(a)$.}
    \label{fig:timing}
\end{figure}
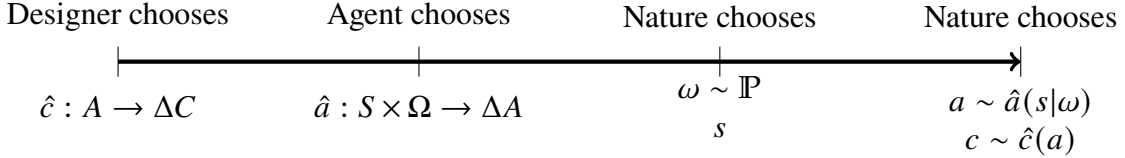

The agent chooses an action $a$ based on the state $s$ and their idiosyncratic type $\omega$. The designer, who does not observe the state, wants to construct consequences $c$ for different actions to incentivize the agent to condition their action on the state $s$. The state $s$ is not payoff-relevant to the agent (to them, $s$ is effectively a private sunspot), while $\omega$ is.

Formally, the designer begins by committing to a mechanism $\hat{c}:\A\rightarrow \Delta\C$ mapping actions to consequences. The agent chooses a strategy $\hat{a}:S\times\Omega\rightarrow \Delta\A$ mapping the state and their idiosyncratic type $\omega$ to actions. Nature then draws the type $\omega$ from a distribution $\mathbb{P}$ and draws the state $s$ arbitrarily\footnote{One may fix a prior distribution $\mu$ over states for communication games. We leave the distribution unspecified for applications to repeated games, where the state reflects agents observations of signals, and hence has an endogeneous distribution.}. The agent's action $a$ is the drawn from $\hat{a}(s\rvert\omega)$, and the consequence $c$ from $\hat{c}(a)$, determining the agent's utility.

The agent's utility for a sequence of play $(\omega,s,a,c)$ is state-independent, given by $u(a,c\rvert\omega)$\footnote{This is in the fashion of the perturbed payoffs of \cite{H73}.}. Our results pertain to the cases where at least one of $\A$ and $\C$ is either finite or finite-dimensional. We assume $(a,c)\mapsto u(a,c\rvert\omega)$ is continuous, and view the utility as a random function drawn from the set of utility functions $\mathbb{U}:=\mathcal{C}^0(\A \times \C )$ --- where $\mathcal{C}^0(X)$ denotes the set of continuous functions  $X\rightarrow\mathbb{R}$.

We say the agent's utility is \textbf{public} if $u(a,c\rvert\omega)\equiv u^0(a,c)$ is type-independent --- in this case the designer can calculate the utility of an outcome to the agent based on $(a,c)$.

The agent's problem is to choose an action strategy $\hat{a}$ that, given $\hat{c}$, solves their \textbf{incentive compatability} constraint: every message sent in equilibrium should maximize their utility
\begin{equation}\label{eq:Abuse.IC}
    \supp{\hat{a}(s\rvert\omega)}\subseteq\argmax_{a\in \A } u(a,\hat{c}(a)\rvert \omega)\qquad \text{for all }s,\omega.
\end{equation}
We will see that if the agent's utility is public, then nearly every strategy can be made incentive compatible through an appropriate choice of $\hat{c}$ --- a fact that has been used by previous papers to tractably study equilibria in this setting --- while it becomes impossible for state-dependent strategies to be incentive compatible in settings with private preferences.

\subsection{Model Extensions}
The equilibria of this model describes a `feasibility set' of outcomes in this setting, by including additional incentives and constraints on the designer we obtain the equilibria in various settings as follows. 

Our model closely describes \textbf{communication games with receiver commitment}. In this case, the agent (sender) observes a state drawn from the prior $\mu\in\Delta S$ and the designer (receiver) has an \textit{ex post} utility function $u_R:\A\times\C\times S\rightarrow \mathbb{R}$. The receiver commits to a mechanism $\hat{c}$ maximizing
\begin{subequations}
\begin{equation}
    \int u_R\left(\hat{a}_{\hat{c}}(s),\hat{c}\circ\hat{a}_{\hat{c}}(s)\rvert s\right)\,d\mu(s),
\end{equation}
where $\hat{a}_{\hat{c}}$ is the strategy chosen by the sender after observing the mechanism $\hat{c}$, constrained to solve eq. \ref{eq:Abuse.IC}. Note that we do not assume the sender breaks ties in favour of the receiver.

In \textbf{communication games} (without commitment), the state is distributed according to a prior $\mu\in\Delta S$, and the receiver is constrained by the \textit{interim} incentive compatability constraint: 
\begin{equation}\label{IC:cheaptalk}
    \supp{\hat{c}(a)}\subseteq \argmax_{c\in \C } \int u_R(a,c\rvert s)\, d\nu_a(s)=:C^\ast(\nu_a,a).
\end{equation}
where $\nu_a\in \Delta S$ the Bayesian posterior belief conditional on the action (`signal') $a$. In general this signal can be payoff relevant to the sender and/or receiver; in \textbf{cheap talk} actions are freely chosen messages that are payoff-relevant to neither player.

\textbf{Mediated} and \textbf{long cheap talk} can also be described in this manner, detailed in Appendix \ref{app:comm.g}.
\end{subequations}
\subsubsection*{Repeated Games} 
In \textbf{repeated games} with $N$ players, each player $i$ chooses an action $a_i\in \A_i$ to take in the current stage game based on the previous signals they have observed $s$. The mechanism $\hat{c}:A\rightarrow \mathbb{R}^N $ describes how the current action profile $a\in \A:=\bigtimes_{i\in N} \A_i$ affects continuation payoffs. In an infinitely repeated game, $\C\subseteq \R^N$ is the set of possible continuation payoffs. With discount rate $\delta$, preferences are described by
$$
    u_i(a;\hat{c}(a))=(1-\delta)v_i(a)+\delta \hat{c}_i(a).
$$
where $v_i$ is the payoff from the stage game and $\delta$ is the discount factor. 

For a set $C^\ast\subseteq \C$ of continuation payoffs, the set of rationalizable action profiles is
\begin{subequations}
\begin{equation}
\begin{split}
    \A^\ast(C^\ast):=\Big\{a^\ast\in \A;a_i^\ast \in \argmax_{a_i\in \A_i} u_i(a_i, a_{-i}^\ast;\hat{c}(a_i,a_{-i}^\ast))\quad \forall i\in N,\exists\hat{c}:A\rightarrow C^\ast\Big\},
\end{split}
\end{equation}
that is, the set of actions such that there is a continuation payoff mapping into $C^\ast$ that solves eq. \ref{eq:Abuse.IC} for every $i$.\footnote{Imperfect monitoring imposes additional constraints on $a^\ast,\hat{c}$. The $a^\ast_{-i}$ must be replaced by the expected action of other agents given the signals observed by agent $i$, while if $G(ds\rvert a)$ is the distribution of signal profiles associated with the action profile $a$, then $\hat{c}(a)$ must be of the form $\int \tilde{c}(s) G(ds\rvert a)$ where $\tilde{c}(s)$ is the continuation payoffs after signal profile $s$. }

With public randomization, the solution concept of subgame perfect equilibria then imposes the additional constraint that the relevant continuation payoffs $C^\ast$ can be attained from action profiles they rationalize $A^\ast(C^\ast)$, formally
\begin{equation}
    C^\ast\subseteq v(\A^\ast(C^\ast)).
\end{equation}

With \textit{perfect monitoring} of actions, where the state (history) is common knowledge, our results will only apply to specific equilibria where agents condition their strategy on disjoint information. The results will have more bite in settings with private information (e.g. imperfect public+private monitoring). The applications to repeated games will be elaborated in detail in Section \ref{sec:app.rg}.

\end{subequations}
\section{Abusing Indifference}\label{sec:thesis}
In all of the above settings, it is desirable for the agent to be able to condition their strategy on the sunspot $s$ --- this enables informative communication in communication games, and credible punishment in repeated games. 

Assuming preferences are public provides immense power to this end. This is because it is sufficient for the agent's problem (the right side of eq. \ref{eq:Abuse.IC}) to have multiple solutions, allowing them to change their action arbitrarily within this set of best responses depending on the state $s$.

Suppose the agent's preference is public, given by $u^0$, and we want to know what strategies are rationalizable when the designer constrains themselves to a compact set of consequences $C_c\subseteq \C$. Clearly, it is never rational for the agent to choose a strictly dominated action, so we restrict attention to strategies supported on the undominated actions $\A_u(C_c)$, formally
$$
    \A_u(C_c):=\Big\{a\in \A  ;\max_{c\in C_c} u^0(a,c)\ge \sup_{a'}\min_{c'\in C_c}u^0(a',c')\Big\}.
$$
The strength of transparent preferences is that this is the \textit{only} constraint on the strategies the designer can induce:
\begin{prop}\label{prop:pre.ex}
    For any compact set $C_c\subseteq\C$, there exists a mechanism $\hat{c}:\A\rightarrow C_c$ that rationalizes any strategy $\hat{a}:S\rightarrow \Delta \A_u(C_c)$.
\end{prop}
Such a mechanism $\hat{c}$ makes the agent indifferent over every undominated action.
\begin{proof}
    We construct a mechanism whereby each action $a\in \A_u(C_c)$ obtains utility
    $$ 
    u^\ast:=\sup_{a'}\min_{c\in C_c} u^0(a',c)
    $$
    Denote $\underline{c}(a)$ the minimizer of $c\mapsto u^0(a,c)$. For all $\epsilon>0$, there exists an action $a_\epsilon\in\A$ such that
    $$
    u^0(a_\epsilon,\underline{c}(a_\epsilon))>u^\ast-\epsilon.
    $$
    Because $a\in \A_u(C_c)$ cannot be dominated by any such $a_\epsilon$, it must be the case that there exists a consequence $\overline{c}(a)$ such that
    $$
    u^0(a,\overline{c}(a))\ge \lim_{\epsilon\rightarrow 0}u^0(a_\epsilon,\underline{c}(a_\epsilon))=u^\ast\ge u^0(a,\underline{c}(a)).
    $$
    Taking a convex combination of $\overline{c}(a)$ and $\underline{c}(a)$, it is then possible to make it so that an agent obtains a utility exactly equal to $u^\ast$ for such an action $a\in \A_u(C_c)$.
\end{proof}
A modest extension of this result in the context of Example \ref{ex:mb} shows how allowing the agent to burn money allows almost total flexibility in the correlation between states and consequences:
\begin{cor}\label{cor:m.b}
        Consider Example \ref{ex:mb}. Let $\hat{c}_s:S\rightarrow \Delta C$ be a map that gives the customer an expected utility of at least their reservation utility $0$. There exists an equilibrium that implements $\hat{c}_s$. 
\end{cor}
The mechanism that obtains this is the same described in the example, and constructed in the previous proof, namely choosing the agent's least preferred option $\underline{c}$ unless they burn a sufficiently high quantity of money to offset their gain from other actions. The agent is then indifferent over a set of messages capable of inducing every action, and can randomize in each state in a way that corresponds to $\hat{c}_s$. This reasoning carries forward to any version of our model that includes money-burning (in some cases requiring designer to commit to randomize after some actions).

Note that while there is an equilibrium that generates the receiver's optimal utility, it is not possible for them to use receiver commitment power to select this equilibrium: for any negotiation protocol it is always a best response for the sender to adopt a constant strategy earning the customer at most $u^\ast_R$.

\subsection{Fragility}\label{sec:antithesis}
The problem with Abusing Indifference is that it relies on the agent having the specific preference that is made indifferent by the proposed mechanism. If the agent has slightly different preferences, they may no longer be indifferent and thus find it optimal to always choose the same action $a$ regardless of the state $s$.

A strategy $\hat{a}:S\times\Omega\rightarrow \A$ is \textbf{uninformative} if the strategy is state-independent for almost every type $\omega$:
$$
    \hat{a}(s\rvert\omega)=\hat{a}(s'\rvert \omega)\qquad\text{ for a.a. $\omega$, for all $s,s'\in S$}.
$$
For an equilibrium strategy to be informative, it is necessary that a positive measure of types can be simultaneously made indifferent over multiple outcomes.

A specific mechanism $\hat{c}$ defines a set of possible outcomes in $A\times \Delta C$. Embedding this in $\Delta\A\times\Delta\C$ and taking the convex hull, we obtain a convex set $\Gamma_{\hat{c}}\subseteq \Delta\A\times\Delta\C$ such that an agent must be indifferent between multiple elements of $\Gamma_{\hat{c}}$ to adopt an informative strategy.\footnote{Formally, $\Gamma_{\hat{c}}:=\{(p,\hat{c}_\#p);p\in\Delta\A\}$, where $\hat{c}_\#p$ is the pushforward of the measure $p$ by $\hat{c}$.} 
\begin{fact}
    If there exists an informative equilibrium then there exists a convex set $\Gamma\subseteq \Delta \A  \times\Delta \C$ such that
\begin{equation}\label{eq:convex.imp}
    \P{\left|\argmax_{(a,c)\in \Gamma} u(a,c\rvert\omega)\right|\ge 2}>0.
\end{equation}
\end{fact}

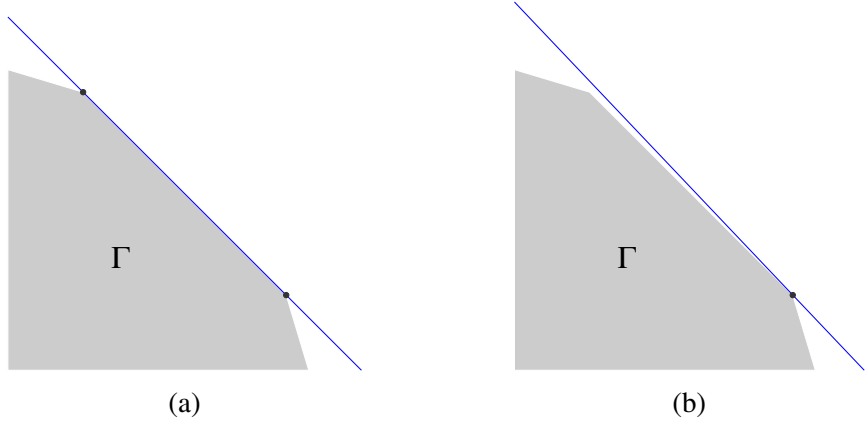
\begin{figure}
    \centering
    \begin{subfigure}{0.4\textwidth}
    \centering
    \begin{tikzpicture}
        \draw[draw =white, fill = black!20] (0,0) -- (4,0) -- (3.7,1) -- (1,3.7)  -- (0,4) -- cycle;
        \draw[blue] (4.7,0) -- (0,4.7);
        \filldraw [black!80] (3.7,1) circle (1 pt);
        \filldraw [black!80] (1,3.7) circle (1 pt);
        \node at (1.5,1.5) {$\Gamma$};
    \end{tikzpicture}
    \caption{}
    \end{subfigure}
    \begin{subfigure}{0.4\textwidth}
    \centering
    \begin{tikzpicture}
        \draw[draw =white, fill = black!20] (0,0) -- (4,0) -- (3.7,1) -- (1,3.7) -- (0,4) -- cycle;
        \draw[blue] (0,4.9) -- (4.65,0);
        \filldraw [black!80] (3.7,1) circle (1 pt);
        \node at (1.5,1.5) {$\Gamma$};
    \end{tikzpicture}
    \caption{}
    \end{subfigure}
    \caption{The information that can be communicated with invariant preferences is limited by the measure of types $\omega$ whose indifference curves are multi-tangent to the same convex set $\Gamma$. Even if this holds for a specific type (left), it will fail for very nearby types (right).}
    \label{fig:convex}
\end{figure}
Specifically, a positive measure of agents must have indifference hyperplanes that are multi-tangent to the same convex set $\Gamma$. While it is easy to construct a set $\Gamma$ that is multi-tangent for a fixed type $\omega$, this will be knife-edge --- agents with very nearby preferences would have only a single tangency point, as illustrated in Figure \ref{fig:convex}.

To formalize the notion that unique maximizers are `common', we use the following notion:
\begin{defn}\label{def:gen}
    A private utility function $u:X\times\Omega\rightarrow \mathbb{R}$ is \textbf{prevalent}\footnote{The term is drawn from \cite[Definition 6]{Hunt92}, which defines prevalent sets as a generalization of generic sets to infinite dimensional spaces. Specifically, if a private preference fails to be prevalent then it lies within a shy subset of $\mathbb{U}$ with positive probability.}
    \begin{itemize}
        \item (for finite $X$) if $u$ is distributed according to a density on $\mathbb{R}^X$. 
        \item (for finite-dimensional $X\subseteq \mathbb{R}^n$) if $u(x\rvert\omega)=\tilde{u}(x\rvert \omega)+\sum_{i=1}^n \lambda_i(\omega)v_i(x\rvert\omega)$ for functions $\tilde{u}:X\times\Omega\rightarrow \mathbb{R},\lambda:\Omega\rightarrow \mathbb{R}^n,v:X\times\Omega\rightarrow \mathbb{R}^n$, such that $v(\cdot\rvert\omega)$ is injective for every $\omega$, and the distribution of $\lambda$ conditional on $\tilde{u},v$ admits a density.
    \end{itemize}
\end{defn}
The simplest version of a prevalent utility in the latter case is when $\tilde{u}$ and $v$ are public and only $\lambda$ is private (and is drawn from a density).\footnote{Considering such a utility on the $n$-dimensional simplex $X=\Delta^n$, with $\tilde{u}\equiv 0$ and $v_i(x)=x_i$, we see that the finite definition is a special case of the finite-dimensional definition (modulo normalizing the utility of $x_{n+1}$).}

These utilities describe scenarios where there is some uncertainty about how an agent values different `dimensions' of the space $X$. In the finite case, these dimensions are given by the elements of $X$ (which can be viewed as vectors in $\Delta X$), in the finite-dimensional space these dimensions are given by the coordinate function $v$. 



\begin{thm}\label{thm:fin.fail}
    For prevalent private utility functions, the only rationalizable strategies are uninformative.

    Moreover, if $\mathbb{U}$ is a separable banach space, then for a fixed mechanism $\hat{c}$ there is a topologically meagre subset $U'\subseteq \mathbb{U}$, such that informative strategies are incentive compatible only if $\mathbb{P}[u\in U']>0$.
\end{thm}
\begin{proof}

Let $\Gamma\subseteq \Delta \A \times \Delta\C$ denote the possible outcomes associated with a mechanism. The agent's objective is to choose the outcome $\gamma\in\Gamma$ that maximizes their expected utility:
\begin{equation}\label{eq:fgamma}
f_\Gamma(u):=\max_{\gamma\in\Gamma} \int u(a,c)\,d\gamma(a,c).
\end{equation}
Note that the agent's value $f_\Gamma$ is convex in their utility $u$, and convex analysis says the solutions $\gamma^\ast$ will lie in the subdifferential $\partial f_\Gamma(u)$ (formally describing the supporting hyperplanes of $f_\Gamma$ at $u$).

For a utility $u$ to have multiple solutions, it must be the case that there are multiple supporting hyperplanes --- ie. $|\partial f(u)|\ge 1$ or $f_\Gamma$ is not Gâteaux differentiable at $u$ (if $\mathbb{U}$ is finite dimensional, this coincides with standard notions of differentiability). However, convex functions are differentiable for most arguments:

Specifically, Rademacher's theorem implies that any convex function $f_\Gamma$ on a finite-dimensional space is differentiable at lebesgue a.e. $u$ (\cite[Theorem 25.5]{Rockafellar}), providing the first prevalent property.

If $\A\times\C\subseteq \mathbb{R}^n$ is finite-dimensional, let $u=\tilde{u}+\sum_{i=1}^n \lambda_i v_i$ as in Definition \ref{def:gen}. Fixing $\tilde{u},v$, define the convex function $g_\Gamma^{\tilde{u},v}:\R^n\rightarrow \R$
$$
g_\Gamma^{\tilde{u},v}(\lambda):=f_\Gamma(\tilde{u}+\lambda_1v_1+\cdots+\lambda_Nv_N).
$$
The subdifferential of this function then contains the $(v_i)_i$ coordinates of the solutions to the agent's problem. Since $g_\Gamma^{\tilde{u},v}$ is differentiable for a.e. $\lambda$, the solution is unique for a.e. $\lambda,\tilde{u},v$.

If $\mathbb{U}$ is a separable Banach space, then for any continuous convex function $f_\Gamma:\mathbb{U}\rightarrow\mathbb{R}$, the set where it is not differentiable $\{u;|\partial f(u)|>1\}$ is meagre (\cite[Theorem 5.61]{convex22}).
\end{proof}
By slightly modifying the proof we can obtain similar results when we only consider uncertainty over one aspect of the agent's preference:

\begin{prop}\label{prop:semi.gen}
    Decompose the agents utility $u(\cdot\rvert \omega)=v_0(\cdot\rvert \omega)+v_1(\cdot\rvert \omega)$ where $v_0$ and $v_1$ are independent.
    \begin{enumerate}[label = (\alph*)]
        \item If $A=A_0\times A_1$ and $v_0:\A_0\times\Omega\rightarrow\mathbb{R}$ is a prevalent private utility function on the space of actions then, decomposing the strategy $\hat{a}:=(\hat{a}_0,\hat{a}_1)$ all rationalizeable $\hat{a}$ have uninformative $\hat{a}_0$.
        \item If $v_0:\C\times\Omega\rightarrow \mathbb{R}$ is a prevalent private utility function on the space of consequences and either (i) $|\C|<\infty$, or (ii) $\C\subseteq \R^N$ is finite-dimensional and the designer is restricted to pure mechanisms $\hat{c}:\A\rightarrow \C$, then the only incentive compatible strategies $\hat{a}$ are such that $\hat{c}\circ \hat{a}:S\rightarrow \Delta\C$ is constant.
    \end{enumerate}
    The results hold without independence of $v_0,v_1$ as long as $v_0$ remains prevalent when conditioned on a.e. $v_1$.
\end{prop}
The first result (a) is relevant if we consider a repeated game where the stage game has agents' idiosyncrasy drawn at the beginning of each period.

The second result (b) is useful in cheap talk models, where messages are not payoff-relevant ($v_1\equiv 0$). We cannot rule out informative strategies in this case. Indeed, if the designer commits to the same choice after every message, every strategy will be incentive compatible for the sender. However, this is the \textit{only} way to generate informative strategies: the designer cannot ever use any of the information that they are able to get from the sender.\footnote{The reason that we need to restrict to pure strategies in (ii) is because we need to form a basis over the set of possible consequences. Thus it cannot be applied to $\Delta\C$ if $|\C|$ is a continuum. This is not a problem in communication game applications if the receiver's best response contains finite pure actions.}
\subsection{Interpreting Abusing Indifference}\label{sec:resolutions}
Since it is irrational to adopt strategies violating III under realistic privacy conditions, results obtained by abusing indifference cannot be interpreted the same way as other equilibria that involve indifference. In this section we advance two related interpretations of these equilibria: the first sees abusing indifference as a powerful, but incomplete, technique for finding realistic equilibria (where realistic merely means compatible with private preferences), the second sees it as a way of obtaining an `upper bound' on the set of realistic equilibrium possibilities. 

\subsubsection*{Using Indifference}
There are two ways to perturb an equilibrium violating III to be compatible with private preferences:
\begin{enumerate}[label = (\arabic*)]
    \item Make the state payoff-relevant to the agent  --- ie. the state is no longer a \textit{sunspot}.
    \item Make the state informative about the mechanism --- ie. the state is no longer \textit{private}.
\end{enumerate}
Beyond justifying the equilibrium in question, introducing these perturbations may also have theoretic benefits: they may
\begin{enumerate}[label = (\alph*)]
    \item demonstrate forces and structures necessary to support desirable equilibrium outcomes,
    \item extend the insights of `Abusing Indifference' to a richer class of models (for example, in communication games, beyond the measure-0 set of precisely state-independent preferences).
\end{enumerate}
As an example of (1), \cite{SGGK23} study how cheap talk communication can occur in finite-action binary-state settings when the sender's preference is slightly state-dependent. Towards the goal of (b), it seems that the obtained equilibria are preserved if the sender's preference is significantly state-dependent, only requiring that the \textit{ordinal} ranking of pure actions is state-independent. 

\cite{B24} considers a similar problem in multi-state settings, as another example of (1). Towards the goal of (a), he finds that aligning preferences\footnote{In the absence of other state-dependence, aligning preferences (ie. making the sender put $\epsilon$ weight on the receiver's preference) is a natural method for introducing state-dependence to sender-receiver games. This can arise if a sender cares slightly about receiver satisfaction (perhaps motivated by reputational concerns), or akin to an sender who breaks ties in favour of a receiver.} is insufficient to preserve cheap talk communication in a wide class of `qualitative' settings: `burning money' is an essential part of credible communication in such cases.

\cite{AGS23} demonstrate (2) in cheap talk games. By furnishing the receiver with a partially informative signal of the state, the sender's information will correlate with the receiver's knowledge, and thus their strategy. It is shown that in certain binary state settings this preserves the unanimously preferred cheap talk equilibrium.

In repeated games with imperfect monitoring, (1) describes an agent with history-dependent preferences, while (2) describes both correlated monitoring (\cite{MM02}) where signals are exogeneously correlated, and belief-based equilibria (e.g. \cite{BO02}) where signals are endogeneously correlated through agents mixed strategies (to be discussed more in the next section).\footnote{As correlated information can often be endogeneously created in repeated games, we need to be careful with our interpretation of our results in these settings. In Section \ref{sec:disc} we argue that it is likely possible to approximate equilibria violating III with belief based equilibria, but these equilibria are \textit{qualitatively} distinct from those they approximate. Moreover, constructing approximating equilibria is highly non-trivial.}

\subsubsection*{An Upper Bound}
A second interpretation of Abusing Indifference equilibria is as an agnostic benchmark for outcomes in nearby models.

An upper-hemicontinuity argument generally shows that for slight perturbations (to agent payoffs/beliefs), the resulting equilibria will approximate equilibria of the Abusing Indifference model, thus the set of equilibria of the unperturbed game form an upper bound (in the set-theoretic sense) on the limit equilibria of nearby games.

In a similar way, Abusing Indifference equilibria may be interpreted as $\epsilon$-equilibria, relaxing incentive constraints so that agents are only constrained to strategies that yield utility within $\epsilon$ of their best response utility ($\epsilon$ should be understood as a proxy for the private variance of preferences). 

This $\epsilon$ relaxation is often interpreted as cover for some small omitted factor in agents' decision-making process. For example, there may be an omitted `status quo bias' that alters agent preferences in a state-dependent manner, akin to perturbation (1),\footnote{The application of status quo bias is clear for pure strategy equilibria --- if the state is $s$, the agent is biased to choose $\hat{a}(s)$ --- but for mixed equilibria it is less obvious: unless the bias can make randomizing a strict best response, the bias must induce the right proportion of agents to choose each action. For some equilibria, it may be necessary for the status quo to be stochastic.} alternatively, these equilibria can be understood as allowing for agents to be slightly irrational. This is less satisfying, as for Abusing Indifference equilibria, it is necessary that this irrationality is expressed in a state-dependent manner --- which in the author's opinion is better described as a state-dependent bias, in the sense of (1).



\section{Applications to Repeated Games}\label{sec:app.rg}

In this section we conceive of repeated games as a general environment where agents make decisions in multiple rounds. This framework includes finitely repeated games, games with random matching between rounds, repetitions of different games, and games will communication between rounds.
\subsection{Model of Repeated Games}

Formally, there is a set of agents $\mathbb{I}$ who make decisions over a set of periods $\mathbb{T}\subseteq \mathbb{N}$. Agent $i$ in period $t$ chooses an action $a_{i,t}$ from the set $A_{i,t}$. We will denote $a_{-i,t}:=(a_{j,t})_{j\neq i}$ the profile of actions of agents other than $i$ in period $t$; $a_{t}:= (a_{i,t})_{i\in \mathbb{I}}$ to indicate the complete profile of actions chosen in period $t$; $a:= (a_{i,t})_{i\in \mathbb{I},t\in \mathbb{T}}$ to be a complete sequence of play; if $t$ is preceded by an inequality it refers to all periods satisfying that inequality, eg. $a_{i,>t}:= (a_{i,\tau})_{\tau>t}$. We use the same shorthands for signals $s$.

The information that agent $i$ learn about the play in the previous period $t$ is described by a signal $s_{i,t}$ drawn from a distribution $G_t(a_t)$. We assume that agents have perfect memory of past signals, and that the signal trajectory $s_{i,\le t}$ is the full information available to agent $i$ at time $t$ (incorporating an agent's memory of their previous actions). The complete signal trajectory $s$ lives in a sequence space $S=\bigtimes_i S_i$ equipped with filtrations $\mathcal{H}_{i,t}$ describing the information about the history known to agent $i$ at the beginning of period $t$. 

We assume that an agent's idiosyncrasy in period $t$, denoted $\omega_t$ and revealed at the beginning of each period, has a transient component $\omega_t^T$ that is independently drawn each period (thus in general $\omega_t\equiv(\omega_t^P,\omega_t^T)$ where $\omega_t^P$ is persistent).

Agent $i$ has discount factor $\delta$ and acts to maximize
\begin{equation}\label{eq:rep.u}
    u_{i,t}(a,\omega_t)\equiv (1-\delta) u_{i,t}^s(a_{ t},\omega_t)+\delta\E{u_{i,t+1}^c(a_{> t},\omega_{t+1})},
\end{equation}
where $u^s_{i,t}$ is the stage game utility from the current period, and the expectation is over future idiosyncrasies and signal observations. For our fragility theorem to apply, we will assume $\omega_t^T\mapsto u_{i,t}^s(\cdot\rvert \omega_t^T,\omega^P_t)$ is a prevalent private utility for a.e. $\omega_t^P$, requiring that $A_{i,t}$ is finite or finite-dimensional for every $i,t$.

Applying Proposition \ref{prop:semi.gen} to this setting then establishes the following result:
\begin{cor}
    In any BPE agents do not condition their action in period $t$ on information unless it is correlated with others' strategies.
\end{cor}
The remainder of this section explores consequences of this corollary. 

We consider a \textbf{pure strategy} for agent $i$ to be a $\mathcal{H}_{i,t}$-adapted process taking pure values in $A_{i,t}$ --- ie. they do not condition their strategy on their private type $\omega_t$.\footnote{Note that equilibria with public randomization are pure, and can occur if the signal contains a public component that is distributed independent of the previous actions.} We thus use \textbf{mixing} to describe both strategies that explicitly randomize (ie. choose actions $a\in\Delta A_{i,t}\setminus A_{i,t}$) in the case of public preferences, and strategies that condition on $\omega$ in the case of private preferences.



We will work with the following menu of monitoring technologies:
\begin{defn}[Monitoring Terminology]
A signal structure is
\begin{itemize}
    \item \textbf{public monitoring} if $s_{i,t}\equiv (s_t^p,a_{i,t})$ where $s_t^p$ is a common (public) signal;
    \item \textbf{private monitoring} if $\{s_{i,t}\}_{i\in \mathbb{I}}$ are conditionally independent given $a_{t}$;
    \item \textbf{noisy} if the signal $s_{t}$ is drawn from distributions $\{G_{t}(a_{-i,t};a_{i,t})\}_{a_{-i,t}}$ that are mutually absolutely continuous for every $a_{i,t}$;
    \item \textbf{deterministic} if the distributions $\{G_t(a_{t})\}_{a_{t}}$ are degenerate;
    \item \textbf{anonymous} in period $t$ if $A_{i,t}\equiv A_t$ across individuals and $G_t$ is invariant to permutations: $G_t((a_{i,t})_i)=G_t((a_{\pi(i),t})_i)$ for any bijection $\pi:\mathbb{I}\leftrightarrow\mathbb{I}$.
\end{itemize}
\end{defn}
We will say a signal is public+private if it contains both a public and a noisy private component.\footnote{Note that there is a middle ground of \textit{correlated} monitoring that is neither public nor private (studied by \cite{MM02}).}


Anonymous monitoring features in oligopoly models (e.g. \cite{APS86}), where firms only observe the price which is determined by the total production.


A sequence of outcomes is \textbf{stage-Nash} if the outcome in each period $t$ is an equilibrium of the stage game $(\mathbb{I},A_{t},u_{t}^s)$. Note that these are the only equilibria possible if future actions are independent of current actions.

To explore how our concepts apply to repeated games, we first construct several equilibria with one-period punishments in a repeated public goods game with public preferences:
\begin{ex}[Public Goods Game]\label{ex:pub.good}
    Consider a $N$-player game where each agent $i$ must decide an amount $a_i\in[0,1]$ to contribute to the public good. The public good provides value $\kappa\sum_i a_i$ to each agent, where $\kappa\in]1/2,1[$. Notably the public good receives zero funding in the Nash equilibrium of the one shot game. Agents discount the future at rate $\delta$.

    The preferences of agent $i$ at time $t$ over future play is then
    $$
    u_{i,t}(a)=(1-\delta)\sum_{\tau=t}^\infty\delta^{\tau-t}\Big(\kappa (a_{i,\tau}+a^{-i}_{\tau})-a_{i,\tau}\Big),
    $$
    where $a^{-i}_{t}$ is the total contribution of other agents in period $t$, which can be decomposed into
    $$
    u_{i,t}(a_t,c)=(1-\delta)\kappa a^{-i}_{t}+(1-\delta)(\kappa-1)a_{i,t}+\delta c,
    $$
    where $c$ is the continuation payoff. Note that the first term is a constant that agent $i$ cannot affect at time $t$, and thus can be neglected. 
    
    The stage utility $u_{i,t}$ is prevalent if $\kappa$ is drawn from a distribution with a density independent across periods --- this represents a private liquidity shock affecting how agents value cash.

    With perfect monitoring, folk theorem says that full investment can be supported whenever $\delta$ is sufficiently close to 1 (specifically $ \delta\ge\tfrac{1-\kappa}{(N-1)\kappa}$), using grim trigger strategies. However grim trigger is severe, resulting in an equilibrium that is fragile and unsustainable in trembling hand or noisy monitoring settings. 

    Instead, we will be interested in equilibria with minimal punishment for deviations (ie. agents are indifferent between on-path and lower contributions), and where punishments last one period.

    
    We can abuse indifference to find equilibria with minimal punishment for off-path play, making high levels of contributions possible off-path. Specifically, given a bijection $\pi:\mathbb{I}\leftrightarrow\mathbb{I}$ with no fixed point, if player $i$ observes the action of player $\pi(i)$, ie. $s_{i,t}=a_{\pi(i),t}$, we define the \textbf{Proportional Response equilibrium} with expected contribution vector $x=(x_i)_i$ to be
\begin{align}
    \label{eq:prop.resp}
    \hat{a}_{i,t}^\text{PR}-x_i=&\alpha_2({s}_{i,t-1}-x_{\pi(i)}),&
    \text{where } \alpha_2 x_{\pi(i)}\le& x_i\le 1-\alpha_2(1-x_{\pi(i)}).
\end{align}
    Thus an agent shirks (from the expected contribution $x_i$) by an amount proportional to how much the observed agent $\pi(i)$ shirked the previous period. The right hand constraint on $x$ ensures that prescribed contributions lie in $[0,1]$, and the proportionality constant is $\alpha_2:=\tfrac{1-\kappa}{\delta\kappa}$. This constant is calibrated to make each agent indifferent over all actions. This equilibrium makes it possible sustain positive contributions whenever $\alpha_2\le1$, equal to the folk theorem condition when $N=2$.

    Due to linearity, these strategies also form equilibria with noisy private monitoring when
\begin{equation}\label{eq:priv.mon}
    s_{i,t}\sim G_i(a_{\pi(i),t})
\end{equation}
    where $G_i(y)$ is a distribution with mean $y$, supported in $[-\epsilon_0,1+\epsilon_1]$ (where $\epsilon_0,\epsilon_1$ are sufficiently small
    ).\footnote{
    The bounds $\epsilon_0,\epsilon_1$ modify the bounds on the expected contributions $x$ in eq. \ref{eq:prop.resp}.
    } 
    
    The strategy can also be modified to describe an equilibrium with noisy (and anonymous) public monitoring,
\begin{equation}\label{eq:pub.mon}
     s_{t}^p\sim G\big(\sum_{i}a_{i,t}\big),
\end{equation}
    where $G(y)$ has mean $y$ and is supported in $[-\epsilon_0,N+\epsilon_1]$. The \textbf{Public Proportional equilibrium} is
\begin{align}
    \label{eq:prop.respN}
    \hat{a}_{i,t}^\text{p}-x_i=&\alpha_N({s}_{t-1}^p-x)&
    \alpha_N (\epsilon_0+x)\le& x_i\le 1-\alpha_N(N+\epsilon_1-x),
\end{align}
    where $x=\sum_{j}x_i$ and the proportionality constant is $\alpha_N:=\tfrac{1-\kappa}{\delta (N-1)\kappa}$. For small $\epsilon_0,\epsilon_1$, such an equilibrium can sustain contributions close to $1$ when $\delta>\frac{1-\kappa}{\kappa}\frac{N}{N-1}$ (for $\kappa<\frac{N}{2N-1}$, no level of patience can sustain this equilibrium). 
    
    The Public Proportional equilibrium requires higher $\delta$ because a deviating agent must also punish when observing a deviation.\footnote{Higher contributions can be sustained by relaxing our constraints of minimal, one-period punishments} This bound can be improved if the deviating agent `atones' mitigating their punishment. Suppose we have deterministic public signal $s_t^p=\sum_{j}a_{j,t}$. We define the `expected total contribution' $x_t$ in period $t$ to be
    \begin{subequations}\label{eq:atone}
    \begin{align}
        x_{1}:=&N&
        x_{t+1}:=&\begin{cases}
            N+(N-1)\alpha_N(s_t^p-x_t)& s_t^p<x_t,\\
            N&\text{else}.
        \end{cases}
    \end{align}
    The \textbf{Atonement equilibrium} is then defined by initially contributing $\hat{a}_{i,1}^A=1$, and thereafter play $\hat{a}^A_{i,t}:=\min\{1,a_{i,t}^A\}$ where
    \begin{align}
        a_{i,t+1}^\text{A}=&
        \begin{cases}
        1+\alpha_N((s_t^p-a_{i,t})-(x_t-a_{i,t}^A))&s_t^p<x_t,\\
        1&\text{else}.\\
        \end{cases} 
    \end{align}
    \end{subequations}
Here agents make the maximum contribution unless there was a shortfall the previous period. Otherwise, they react only to the shortfall attributable to other players, shading down their contribution a proportional amount --- thus if they are the only deviating agent they contribute the full amount. This sustains full contributions when $\delta\ge \frac{1-\kappa}{\kappa}$ (equal to the constraint for the proportional response equilibrium with perfect monitoring, equal to the folk theorem bound when $N=2$).

This strategy can be modified to accommodate noisy signals, but these adjustments are non-trivial due to the non-linearity of the strategy. 

Atonement allows agents to `consent' to their punishment (possibly by allowing their stage-payoff to go below their maxmin payoff) to allow the quick re-establishment of trust. This also allows continuation payoffs to stay close to the Pareto frontier. With two players, the equilibrium $\hat{a}^A$ actually moves the continuation payoff around the Pareto frontier as one player deviates, illustrated by Figure \ref{fig:folk}.
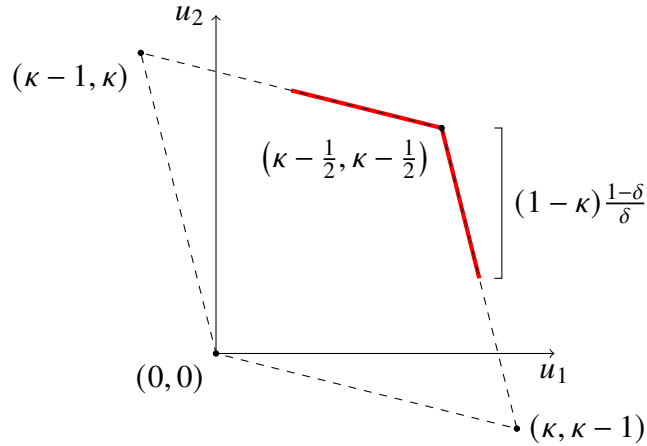
\begin{figure}
    \centering
    \begin{tikzpicture}
        \draw [red, ultra thick] (3.5,1) -- (3,3) -- (1,3.5);
        \draw [dashed, black] (3,3) -- (4,-1) -- (0,0) -- (-1,4) -- cycle;
        \draw[<->] (0,4.5) -- (0,0) -- (4.5,0);
        \node[anchor = north] at (4.5,0) {$u_1$};
        \node[anchor = east] at (0,4.5) {$u_2$};
        \draw (3.7,1) -- (3.8,1) -- (3.8,3) -- (3.7,3);
        \node[anchor = west] at (3.8,2) {$(1-\kappa)\frac{1-\delta}{\delta}$};
        \filldraw [black] (3,3) circle (1 pt) node[anchor = north east] {\color{black} $\big(\kappa-\tfrac{1}{2},\kappa-\tfrac{1}{2}\big)$};
        \filldraw [black] (4,-1) circle (1 pt) node[anchor = west] {\color{black} $(\kappa,\kappa-1)$};
        \filldraw [black] (-1,4) circle (1 pt) node[anchor = north east] {\color{black} $(\kappa-1,\kappa)$};
        \filldraw [black] (0,0) circle (1 pt) node[anchor = north east] {\color{black} $(0,0)$};
    \end{tikzpicture}
    \caption{The dashed line indicates the feasible payoffs of our public goods game. The red line describes the continuation payoff set in our Atonement equilibrium with two players, where a deviation results in a proportional loss to one's own continuation utility, and a proportional gain to the other's.}
    \label{fig:folk}
\end{figure}

\end{ex}

There are three types of equilbiria that we will consider, corresponding to the inferential Directed Acyclic Graphs (DAGs)\footnote{These are slightly different from causal DAGs, where past actions have a causal effect on signals. If these past actions are never chosen, this causal link is omitted from the inferential DAG, as prior beliefs put zero weight on these specific past actions and hence no inference can be made about these actions.} in Figure \ref{fig:dags}: perfect public equilibria (PPE), belief-free equilibria (BFE), and belief-based equilibria (BBE).
\begin{figure}
    \centering
    \begin{subfigure}{0.27\linewidth}\centering
    \begin{tikzpicture}[scale = 0.9,
    roundnode/.style={circle, draw = black, thick, minimum size=14mm},
    squarenode/.style={draw = black, thick, minimum size=12mm},]
        \node[roundnode]  (s) at (0,0) {\small $s^p_{\le t}$};
        \node[roundnode]  (s1) at (0,2) {\small $s_{i,\le t}$};
        \node[squarenode]  (a1) at (2.8,2) {\small $a_{i,>t}$};
        \node[roundnode]  (s2) at (0,-2) {\small $s_{j,\le t}$};
        \node[squarenode]  (a2) at (2.8,-2) {\small $a_{j,>t}$};
        \draw [->,ultra thick] (s) -- (a1);
        \draw [->,ultra thick] (s) -- (a2);
    \end{tikzpicture}
    \renewcommand\thesubfigure{PPE}\caption{}\label{fig:PPE}
    \end{subfigure}
    \begin{subfigure}{0.27\linewidth}\centering
    \begin{tikzpicture}[scale = 0.9,
    roundnode/.style={circle, draw = black, thick, minimum size=14mm},
    squarenode/.style={draw = black, thick, minimum size=12mm},]
        \node[roundnode]  (s) at (0,0) {\small $s^p_{\le t}$};
        \node[roundnode]  (s1) at (0,2) {\small $s_{i,\le t}$};
        \node[roundnode]  (a1) at (2.8,2) {\small $a_{i,>t}$};
        \node[roundnode]  (s2) at (0,-2) {\small $s_{j,\le t}$};
        \node[roundnode]  (a2) at (2.8,-2) {\small $a_{j,>t}$};
        \draw [->,ultra thick] (s1) -- (a1);
        \draw [->,ultra thick] (s) -- (a1);
        \draw [->,ultra thick] (s2) -- (a2);
        \draw [->,ultra thick] (s) -- (a2);
    \end{tikzpicture}
    \renewcommand\thesubfigure{BFE}
    \caption{}\label{fig:BFE}
    \end{subfigure}
    \begin{subfigure}{0.44\linewidth}\centering
    \begin{tikzpicture}[scale = 0.9,
    roundnode/.style={circle, draw = black, thick, minimum size=14mm},
    squarenode/.style={draw = black, thick, minimum size=12mm},]
        \node[roundnode]  (a0) at (-2.8,2) {\small $a_{i,\le t}$};
        \node[roundnode]  (a1) at (-2.8,-2) {\small $a_{j,\le t}$};
        \node[roundnode]  (s) at (0,0) {\small $s^p_{\le t}$};
        \node[roundnode]  (s1) at (0,2) {\small $s_{i,\le t}$};
        \node[roundnode]  (s2) at (0,-2) {\small $s_{j,\le t}$};
        \node[squarenode]  (a2) at (2.8,2) {\small $a_{i,>t}$};
        \node[squarenode]  (a3) at (2.8,-2) {\small $a_{j,>t}$};
        \draw [->,ultra thick] (a0) -- (s);
        \draw [->,ultra thick] (a0) -- (s1);
        \draw [->,ultra thick] (a0) -- (s2);
        \draw [->,ultra thick] (a1) -- (s);
        \draw [->,ultra thick] (a1) -- (s1);
        \draw [->,ultra thick] (a1) -- (s2);
        \draw [->,ultra thick] (s1) -- (a2);
        \draw [->,ultra thick] (s) -- (a2);
        \draw [->,ultra thick] (s2) -- (a3);
        \draw [->,ultra thick] (s) -- (a3);
    \end{tikzpicture}
    \renewcommand\thesubfigure{BBE}
    \caption{}\label{fig:BBE}
    \end{subfigure}
    \caption{Inferential DAGs relating past private ($s_{i,\le t}$) and public ($s^p_{\le t}$) signals to strategies ($a_{i,>t}$), square nodes are payoff-relevant inferences. In PPE public signals form a coordination device and private information is irrelevant; in BFE agents condition on signals that do not carry payoff-relevant information; in BBE signals are informative about others' past actions ($a_{j,\le t}$) and, through that, their future strategy.}
    \label{fig:dags}
\end{figure}
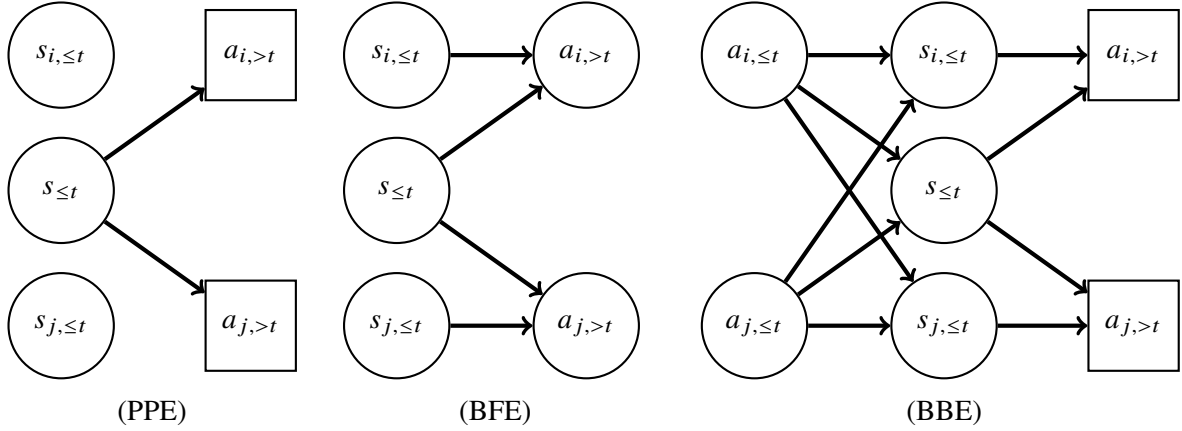

\subsection{Public Perfect Equilibria}
Let $\mathcal{K}_t:=\bigcap_{i\in I}\mathcal{H}_{i,t}$ be the common knowledge filtration. A \textbf{perfect public equilibrium} (\textbf{PPE}) is then a perfect bayesian equilibrium (PBE) such that all strategies are $\mathcal{K}_t$-adapted. An example is the Public Proportional equilibrium, while the Atonement equilibrium is not a PPE, since agents condition on their own previous action $a_{i,t-1}$ which cannot be deduced by other agents (when $N>2$).

Let $\Sigma_{i,t}\subseteq \mathcal{H}_{i,t}$ be the $\sigma$-algebra of information available to agent $i$ at time $t$ that is relevant to their strategy, likewise let $\Sigma_{-i,t}:=\bigvee_{j\neq i}\Sigma_{j,t}$ be the $\sigma$-algebra relevant to other agents.\footnote{Formally, $\Sigma_{i,t}:=\mathcal{H}_{i,t}\cap \sigma(\hat{a}_{i})$, where $\sigma(\hat{a}_{i})$ is the smallest $\sigma$-algebra such that $\hat{a}_{i,t}$ is $\sigma(\hat{a}_{i})$-measurable for all $t$. The notation $\bigvee_{j} \Sigma_{j}$ is defined as $\sigma(\bigcup_{j} \Sigma_j)$, ie. the smallest $\sigma$-algebra that contains every $\Sigma_j$.}

Our main corollary is that private information has \textit{no} value in pure equilibria:
\begin{cor}[The Irrelevance of Private Information]\label{cor:pub.mon}
    Consider a repeated game with prevalent private preferences.
    \begin{enumerate}
        \item With noisy public+private monitoring, any pure PBE is a PPE.
        \item In any PPE $\Sigma_{i,t}\subseteq \Sigma_{-i,t}$.
        \item With noisy private monitoring, any outcome that occurs w.p. 1 in a given period $t$ is stage-Nash.
    \end{enumerate}
\end{cor}
The first statement says that if we restrict attention to pure strategy equilibria, we are restricted to perfect public equilibria.

The second statement says that even in PPE, agents can only react to information that other agents also react to. Thus even with perfect monitoring (where SPE and PPE coincide), the Proportional Response equilibrium cannot be implemented, as each agent $i$ conditions their action on a distinct signal $a_{\pi(i),t}$. This doesn't apply to the Atonement equilibrium, as there agent $i$ conditions their strategy on $s_t^p,a_{i,t}$ where $a_{i,t}=s_t^p-\sum_{j\neq i}a_{j,t}$ can be deduced from the information relevant to other agents' strategies.

The last statement says that even if we consider mixed continuation strategies, the only attainable, pure stage-outcomes with private monitoring are stage-Nash.

\begin{proof}[Proof of Corollary \ref{cor:pub.mon}]
    1. Consider two histories that differ in private signal but not public signal. Since strategies are pure (and signalling is noisy), the signal is not informative about others' actions the previous period, and because signals are private, this means the signal is not informative about others' signals. Thus an agent's belief about others' strategies is necessarily independent of the private signal they observe, which is then a private sunspot.

    2. Similarly, if two public signal paths $s_{\le t},s_{\le t}'$ result in the same strategy for agents $j\neq i$, then $\{s_{\le t},s_{\le t}'\}$ is effectively a private sunspot for agent $i$: if they conditioned their strategy on the realization of these two paths, a designer could use these sunspot-independent strategies of other agents to induce an agent to condition on the private sunspot.  

    3. If agents always take an action $a_t$ w.p. 1 in period $t$ in a setting with noisy private monitoring, then $s_{i,t}$ will be independent across agents and consequently a private sunspot. Applying Proposition \ref{prop:semi.gen}, agents cannot condition their future strategy on $s_{i,t}$, and $a_{i,t}$ must be myopically optimal for each agent, resulting in a stage-Nash equilibrium.
\end{proof}


    

This justifies PPE in one sense, but PPE also often have unsatisfying properties: agents cannot react to private information that incriminates others, and in the absence of an informative public signal the only PPE are stage-Nash. 

Moreover, in the case of anonymous monitoring it becomes more difficult to obtain equilibria that have atonement or are renegotiation proof. For the following definition, suppose monitoring is public so $s_{i,t}=(s_t^p,a_{i,t})$ and that continuation payoffs for period $t$ can be written as an expectation over the signals observed at the start of period $t+1$.
\begin{defn}
    Fix an equilibrium, and let $v^c_{i,t}:S\rightarrow\mathbb{R}$ be the expected continuation payoff of player $i$ in period $t$ conditional on their information $\mathcal{H}_{i,t+1}$ obtained at the end of the period. The equilibrium features \textbf{atonement} if, holding the public signal fixed, an agent obtains a worse continuation payoff if they chose a myopically better action in period $t$. Formally, if $a'_{k,t'}=a_{k,t'}$ for all $t'\le t$, $k\neq i$, and $u_{i,t}^s(a'_{t})>u_{i,t}^s(a_{t})$ then
    $$
    v_{i,t}^c(a',s^p)\le v_{i,t}^c(a,s^p)\quad \text{for all $s^p$, strict at least once.}
    $$
    The equilibrium is (weakly) \textbf{renegotiation proof} at period $t$ if there are no two public signal paths $s^p,\tilde{s}^p$ such that
    $$
    v_{i,t}^c(a,\tilde{s}^p)\le v_{i,t}^c(a,s^p) \quad\text{for all $i,a$, strict at least once.}
    $$
\end{defn}

Atonement is a way for a guilty party to consent to punishment, which can then occur more immediately and over fewer periods, allowing trust to be re-established quickly after a defection (this is what allows our Atonement equilibrium (eq. \ref{eq:atone}) to have one period punishments with higher discount factors than our Public Proportional equilibrium). This is particularly desirable in settings where there might be miscoordination and learning of equilibria.

The idea behind renegotiation proof equilibria (developed by \cite{FM89}) is that punishment threats may lose credibility if they make everyone worse off: players would unanimously prefer to renegotiate the path of play to the pareto-improving path induced by the signal $s^p_{\le t}$.


\begin{fact}\label{prop:reneg}
    PPE cannot feature atonement.
    
    With anonymous public monitoring the only symmetric renegotiation proof PPE are stage-Nash.
\end{fact}
Asymmetric renegotiation proof equilibria either randomize who gets punished, or feature agents taking turns to `free-ride' each period, rewarding them with more free-riding if the other agents falter. Both of these alterations come at a significant cost to efficiency.

\subsection{Belief-Free Equilibria}
The Proportional Response equilibria generated in Example \ref{ex:pub.good} is a special case of a Belief Free Equilibrium. These are equilibria where agents are willing to obey the strategy prescribed by their signals regardless of which signals other agents have seen. 

Denote $\hat{a}_i(s_{\le t})$ to be the restriction of the strategy $\hat{a}_i$ to sequences of observations that follow $s_{i,\le t}$.
\begin{defn}
    A \textbf{$T$-block belief-free equilibrium} is a profile of strategies $\hat{a}$ and a set of periods $T\subseteq \mathbb{T}$ such that for any period $t\in T$ and any profile of signal paths $(s_{i,\le t})_{i\in \mathbb{I}}$, the strategy $\hat{a}_i(s_{i,\le t})$ is a best response to $\hat{a}_{-i}(s_{-i,\le t})$.
\end{defn}
The set $T$ describes the initial periods of `blocks' where individuals may have strict best responses depending on their signal within a block, but their strategy at the beginning of each block their best response is independent of others' signals (thus the beliefs about others' signals). 

If $T=\mathbb{T}$ then every period is its own block, describing the original \textit{belief-free equilibria} analyzed by \cite{EHO05}. While the Proportional Response equilibrium (eq. \ref{eq:prop.resp}) is belief free, the Atonement equilibrium is \textit{not} belief free: full contribution is the prescribed action when one is the only defector the previous period, but is sub-optimal if there was a defector other than oneself.

Note that the solution concept of belief-free equilibrium itself invokes $s_{i,\le t}$ as a private sunspot: in our basic mechanism problem, a designer could use a continuation payoff map generated by any $s_{-i,\le t}$ to incentivize agents to condition their strategy on $s_{i,\le t}$.
\begin{cor}\label{cor:BFE}
    With prevalent private preferences, $T$-block belief-free equilibria involve stage-Nash outcomes for every period in $T$.
\end{cor}
This does say that belief free equilibria obtained with public preferences cannot be purified --- it may be possible to still approximate the equilibria using nearby strategies (necessarily involving private randomization) --- but the purifications will not satisfy the belief free property: agents will care about the beliefs of others.

\section{An Informal Analysis of Belief-Based Equilibria}\label{sec:disc}

In the previous section we find several negative properties of PPE in repeated games with imperfect monitoring, and show that these are all the equilibria that are obtained without (private) randomization. Adding private randomization helps by introducing \textit{reputation} to the setting, as past actions can be indicative of their future strategy, introducing informativeness to private signals. This is captured by belief-based equilibria (diagrammed in Figure \ref{fig:BBE}):

\begin{defn}
A \textbf{belief-based equilibrium} of a repeated game (with 2-players or public monitoring) is a PBE such that, in the first period where an agent does not play their myopic best response w.p. 1,
\begin{enumerate}
    \item the agent randomizes (ie. chooses different actions depending on $\omega$),
    \item the agent's action in that period affects their own future strategy.
\end{enumerate}
\end{defn} 
The first condition allows agents to make inferences about previous actions based on the observed signals, while (2) ensures these inferences are payoff-relevant through the inferred effect on others' strategies (otherwise signals remain a sunspot).\footnote{In multi-player games with private monitoring, these inferences may be payoff relevant through another channel, relaxing (2). In particular, they may contain information about a third party's signal, effectively serving to coordinate punishment when the first agent chooses a lower action, loosely ressembling an imperfectly correlated PPE.}
 
Note that these equilibria can be constructed with deterministic (or even perfect) monitoring and public preferences.\footnote{In this case, they can overlap with belief-free equilibria: in settings with binary actions and public preferences, any belief-based equilibrium where agents randomize after every history is necessarily belief-free. However, the structure of incentives will vary across histories (being distributed differently between stage- and continuation-payoffs) in such a way that induces agents with private preferences to discriminate between these histories.}

\cite{BO02} demonstrate a belief-based equilibrium with imperfect monitoring in a repeated prisoner's dilemma. A simple example in our public goods game with (deterministic) public monitoring is a strategy that each period randomizes between contributing $0$ (w.p. $\rho$) and following an atonement strategy (with proportionality $\alpha=\frac{\alpha_N}{1-\rho}$) which prescribes full contribution unless a low signal is observed and the agent did not deviate in the previous period, in which case $1-\alpha$ is prescribed. 

A full analysis of a belief-based equilibrium with both private preferences and noisy monitoring is beyond the scope of this paper, nevertheless I informally describe properties such an equilibrium should have, to illustrate how private preferences and noisy monitoring affect the possible equilibrium levels of randomization. Specifically, suppose preferences are private so that an agent's `liquidity shock' $\kappa$ is drawn from a distribution $F_\kappa$ with density over $[\underline{\kappa},\overline{\kappa}]$ independently each period.

In the proposed equilibrium, low signals harms parties' reputations. Adding private preferences, randomization can be attributed to liquidity shocks: agents whose $\kappa$ lies below a certain (possibly history-dependent) threshold $\kappa^\ast=F_\kappa^{-1}(\rho)$ are willing to sacrifice reputation for a larger payoff today, while those with higher $\kappa$ are willing to contribute the prescribed amount to maintain (or restore) a higher reputation.

To see the relation between heterogeneity in the private preference, randomization, and noisy monitoring, note that if another party failed to contribute the previous period (and a negative signal is realized) only contributions up to $1-\alpha$ are incentivized as the guilty party will make up the difference to restore the equilibrium. However, if the signal is a false negative, and there is no guilty party, contributions up to $1$ are incentivized. Agents at the threshold $\kappa^\ast$ will then have a strict preference to contribute $1-\alpha$ rather than $1$ whenever there is positive probability that an agent failed to contribute the previous period, but agents with higher $\kappa$ will prefer to contribute $1$ unless they have sufficiently high confidence in the negative signal. 

Higher confidence can be obtained by increasing the randomization $\rho$, increasing the prior that other agents are deviating. More randomization is required when 
\begin{enumerate}[label=(\arabic*)]
    \item monitoring is noisier, to counter-balance the increasing likelihood of false negatives;
    \item $\overline{\kappa}-\kappa^\ast$ is higher, to increase confidence in negative signals countering the increased willingness of high $\kappa$ types to deviate to 1.
\end{enumerate}
Recalling that $\rho=F_{\kappa}(\kappa^\ast)$, the second case describes how more uncertainty about $\kappa$ (ie. more disperse distributions $F_\kappa$) requires more randomness $\rho$.

This logic results in an intuitive effect: imprecise monitoring leads to more lenient punishments in expectation. Specifically, as monitoring becomes noisier/less precise, we need more randomization $\rho$ to maintain the `signal-to-noise ratio'. This requires an increase in the threshold $\kappa^\ast$ where agents are willing to defect, which can only be incentivized if the equilibrium involves more lenient punishments. This effect is contrary to many models without heterogeneity (e.g. belief free equilibria, mechanism design), where expected punishment is calibrated to generate indifference and thus does not vary with signal accuracy. In particular, the punishment conditional on a negative signal necessarily increases as the probability of generating a false positive increases.

Conversely, if monitoring becomes arbitrarily accurate (for a fixed distribution $F_\kappa$), less randomization $\rho$ is required: the threshold $\kappa^\ast$ may decrease, eventually approaching $\min\supp{F_\kappa}$. In this case, punishments become more severe, involving contributions close to $0$. The modal $\kappa$ may find these incentives extremely strict --- modelling the game from their perspective this equilibrium resembles a (trembling hand) SPE more than any noisy monitoring refinement.


\printbibliography

@article{M91,
title = {On the theory of repeated games with private information: Part I: anti-folk theorem without communication},
journal = {Economics Letters},
volume = {35},
number = {3},
pages = {253-256},
year = {1991},
issn = {0165-1765},
doi = {https://doi.org/10.1016/0165-1765(91)90139-C},
url = {https://www.sciencedirect.com/science/article/pii/016517659190139C},
author = {Hitoshi Matsushima}
}

@article{APS86,
title = {Optimal cartel equilibria with imperfect monitoring},
journal = {Journal of Economic Theory},
volume = {39},
number = {1},
pages = {251-269},
year = {1986},
issn = {0022-0531},
doi = {https://doi.org/10.1016/0022-0531(86)90028-1},
url = {https://www.sciencedirect.com/science/article/pii/0022053186900281},
author = {Dilip Abreu and David Pearce and Ennio Stacchetti}
}

@article{MM02,
title = {Repeated Games with Almost-Public Monitoring},
journal = {Journal of Economic Theory},
volume = {102},
number = {1},
pages = {189-228},
year = {2002},
issn = {0022-0531},
doi = {https://doi.org/10.1006/jeth.2001.2869},
url = {https://www.sciencedirect.com/science/article/pii/S0022053101928698},
author = {George J. Mailath and Stephen Morris}
}

@article{BO02,
title = {Belief-Based Equilibria in the Repeated Prisoners' Dilemma with Private Monitoring},
journal = {Journal of Economic Theory},
volume = {102},
number = {1},
pages = {40-69},
year = {2002},
issn = {0022-0531},
doi = {https://doi.org/10.1006/jeth.2001.2878},
url = {https://www.sciencedirect.com/science/article/pii/S0022053101928789},
author = {Bhaskar, Venkataraman and Ichiro Obara}
}

@inproceedings{CD23,
author = {Corrao, Roberto and Dai, Yifan},
title = {Mediated Communication with Transparent Motives},
year = {2023},
isbn = {9798400701047},
publisher = {Association for Computing Machinery},
address = {New York, NY, USA},
url = {https://doi.org/10.1145/3580507.3597808},
doi = {10.1145/3580507.3597808},
booktitle = {Proceedings of the 24th ACM Conference on Economics and Computation},
pages = {489},
numpages = {1},
location = {London, United Kingdom},
series = {EC '23}
}

@article{HL25,
title = {The optimality of (stochastic) veto delegation},
journal = {Games and Economic Behavior},
volume = {150},
pages = {215-234},
year = {2025},
issn = {0899-8256},
doi = {https://doi.org/10.1016/j.geb.2024.12.003},
url = {https://www.sciencedirect.com/science/article/pii/S089982562400191X},
author = {Xiaoxiao Hu and Haoran Lei}
}

@misc{Hunt92,
      title={Prevalence: a translation-invariant ``almost every'' on infinite-dimensional spaces}, 
      author={Brian R. Hunt},
      year={1992},
      eprint={math/9210220},
      archivePrefix={arXiv},
      primaryClass={math.FA},
      url={https://arxiv.org/abs/math/9210220}, 
}

@article{EHO05,
author = {Ely, Jeffrey C. and Hörner, Johannes and Olszewski, Wojciech},
title = {Belief-Free Equilibria in Repeated Games},
journal = {Econometrica},
volume = {73},
number = {2},
pages = {377-415},
doi = {https://doi.org/10.1111/j.1468-0262.2005.00583.x},
url = {https://onlinelibrary.wiley.com/doi/abs/10.1111/j.1468-0262.2005.00583.x},
eprint = {https://onlinelibrary.wiley.com/doi/pdf/10.1111/j.1468-0262.2005.00583.x},
year = {2005}
}

@article{H73,
	Author = {Harsanyi, John C. },
	Da = {1973/12/01},
	Id = {Harsanyi1973},
	Isbn = {1432-1270},
	Journal = {International Journal of Game Theory},
	Number = {1},
	Pages = {1--23},
	Title = {Games with randomly disturbed payoffs: A new rationale for mixed-strategy equilibrium points},
	Ty = {JOUR},
	Url = {https://doi.org/10.1007/BF01737554},
	Volume = {2},
	Year = {1973}}

@article{LR20,
author = {Lipnowski, Elliot and Ravid, Doron},
title = {Cheap Talk With Transparent Motives},
journal = {Econometrica},
volume = {88},
number = {4},
pages = {1631-1660},
url = {https://onlinelibrary.wiley.com/doi/abs/10.3982/ECTA15674},
year = {2020}
}

@article{CH10,
Author = {Chakraborty, Archishman and Harbaugh, Rick},
Title = {Persuasion by Cheap Talk},
Journal = {American Economic Review},
Volume = {100},
Number = {5},
Year = {2010},
Month = {12},
Pages = {2361-82},
URL = {https://www.aeaweb.org/articles?id=10.1257/aer.100.5.2361}}

@misc{SGGK23,
      title={Robust equilibria in cheap-talk games with fairly transparent motives}, 
      author={Jan-Henrik Steg and Elshan Garashli and Michael Greinecker and Christoph Kuzmics},
      year={2023},
      eprint={2309.04193},
      archivePrefix={arXiv},
      primaryClass={econ.TH}
}

@misc{AGS23,
      title={Informationally Robust Cheap-Talk}, 
      author={Itai Arieli and Ronen Gradwohl and Rann Smorodinsky},
      year={2023},
      eprint={2302.00281},
      archivePrefix={arXiv},
}

@article{DK21,
	Author = {Diehl, Christoph and Kuzmics, Christoph},
	Da = {2021/12/01},
	Doi = {10.1007/s00182-021-00774-0},
	Id = {Diehl2021},
	Isbn = {1432-1270},
	Journal = {International Journal of Game Theory},
	Number = {4},
	Pages = {911--925},
	Title = {The (non-)robustness of influential cheap talk equilibria when the sender's preferences are state independent},
	Ty = {JOUR},
	Url = {https://doi.org/10.1007/s00182-021-00774-0},
	Volume = {50},
	Year = {2021}}

@misc{B24,
      title={Robust Communication Between Parties with Nearly Independent Preferences}, 
      author={Alistair Barton},
      year={2024},
      eprint={2403.13983},
      archivePrefix={arXiv},
      primaryClass={econ.TH},
      url={https://arxiv.org/abs/2403.13983}, 
}

@techreport{BMM09,
title = {A Foundation for Markov Equilibria in Infinite Horizon Perfect Information Games
},
institution = {Penn Institute for Economic Research},
journal = {PIER Working Paper},
number = {09-029},
year = {2009},
doi = {https://doi.org/10.2139/ssrn.1460780},
url = {https://ssrn.com/abstract=1460780},
author = {Bhaskar, Venkataraman and Mailath, George J. and Morris, Stephen},
series  = {PIER Working Paper},
type    = {Working Paper},
number  = {09-029},
}

@article{AH03,
 ISSN = {00129682, 14680262},
 URL = {http://www.jstor.org/stable/1555534},
 author = {Robert J. Aumann and Sergiu Hart},
 journal = {Econometrica},
 number = {6},
 pages = {1619--1660},
 publisher = {[Wiley, Econometric Society]},
 title = {Long Cheap Talk},
 urldate = {2025-04-30},
 volume = {71},
 year = {2003}
}

@article{FM89,
title = {Renegotiation in repeated games},
journal = {Games and Economic Behavior},
volume = {1},
number = {4},
pages = {327-360},
year = {1989},
issn = {0899-8256},
doi = {https://doi.org/10.1016/0899-8256(89)90021-3},
url = {https://www.sciencedirect.com/science/article/pii/0899825689900213},
author = {Joseph Farrell and Eric Maskin}
}

@article{BMM08,
title = {Purification in the infinitely-repeated prisoners' dilemma},
journal = {Review of Economic Dynamics},
volume = {11},
number = {3},
pages = {515-528},
year = {2008},
issn = {1094-2025},
doi = {https://doi.org/10.1016/j.red.2007.10.004},
url = {https://www.sciencedirect.com/science/article/pii/S1094202507000701},
author = {Bhaskar, Venkataraman and Mailath, George J. and Morris, Stephen}
}

@book{Rockafellar,
author="R. Tyrell Rockafellar",
title="Convex Analysis",
year="1972",
publisher="Princeton University Press",
pages="246",
isbn="978-0-691-01586-6",}

@book{convex22,
author="Mordukhovich, Boris S.
and Mau Nam, Nguyen",
title="Convex Analysis and Beyond: Volume I: Basic Theory",
year="2022",
publisher="Springer International Publishing",
address="Cham",
pages="355",
isbn="978-3-030-94785-9",
doi="10.1007/978-3-030-94785-9_5",
url="https://doi.org/10.1007/978-3-030-94785-9_5"
}
\appendix\section{Proofs}

\begin{proof}[Proof of Corollary \ref{cor:m.b}]
        Suppose the customer commits to an `off-path' mechanism $\hat{a}'$. The contractor can then credibly threaten to only ever send a single message that maximizes their payoff, reducing the customer to (at most) their reservation utility $u^\ast_R$. Such a response can incentivize the customer to commit to any mechanism that obtains utility at least $u^\ast_R$.

        Let $A_c:=\supp{\hat{a}_s}$ be the set of actions that can be chosen by $\hat{a}_s$, and $\overline{\pi}:=\max_{a\in A_c}\pi(a)$. The set of undominated messages is then
        $$
        \M(A_c):=\left\{(m_1,m_2);m_2\le \overline{\pi}\right\}.
        $$
        A mechanism that makes the sender indifferent over $\M(A_c)$ involves choosing actions
        $$
        \hat{a}(m_1,m_2)\in \{a\in \Delta A_c;\pi(a)=m_2\}.
        $$
        This makes it possible for $\hat{a}$ to be surjective on $A_c$ and the contractor can then choose messages $\hat{m}(s\rvert\omega)$ so that the distribution of actions is $\int \hat{a}\circ\hat{m}(s\rvert\omega)\,d\mathbb{P}(\omega)=\hat{a}_s(s)$.
    \end{proof}

\section{Other communication games}\label{app:comm.g}

In \textbf{mediated cheap talk} the sender sends a cheap talk message $a$ revealing the state to a mediator who then designs an experiment $\eta\in\Delta\Delta S$ to reveal information to the receiver, thus $\A=S$. Define $\eta_a(\nu):=\eta(\nu)\frac{d\nu(a)}{d\mu(a)}$ to be the Bayes measure over posteriors $\nu$ when the state is claimed to be $a$. Our `designer' is an aggregate of the mediator and receiver's strategies, constrained by
\begin{align}
    \hat{c}(a)=&c^\ast_\#(\eta_a)& 
    c^\ast(\nu)\in& C^\ast(\nu)&
    \mu=&\int \nu \,d\eta(\nu)
\end{align}
where $c^\ast_\#(\eta_a)$ is the pushfordward of $\eta_a$ by $c^\ast$ (ie. the distribution of choices when beliefs $\nu$ are drawn from $\eta_a$ and choices made according to $c^\ast:\nu\mapsto c$), and $\mu\in \Delta S$ is the prior belief. $C^\ast(\nu)$ is the best response to posterior $\nu$, defined in eq. \ref{IC:cheaptalk}. The last equation describes the set of experiments $\eta$ satisfying Bayes' Law.

In \textbf{long cheap talk} games, there is a sequence of communication opportunities alternating with random public signals which may affect the interpretation of messages. This can be seen as a recursive version of our design problem where, if the receiver enters round $t$ with prior belief $\nu_{t-1}$ and the coordination device generates public signal $\lambda_t\in[0,1]$ the sender sends messages that induce the posterior $\nu_{t}$ with probability $\eta_{t}(\nu_{t}\rvert\nu_{t-1},\lambda_t)$, where $\nu_0=\mu$ is the prior. The belief $\nu_t$ then evolves according to a stochastic process (martingale), which we constrain to converge a.s. (with finite $S$ this is guaranteed by the martingale convergence theorem). Define $\eta_t^\infty(\nu\rvert\lambda)$ to be distribution of $\lim \nu_{\tau}$ conditional on $\nu_t=\nu$ and $\lambda_t=\lambda$. The designer in stage $t$ is then constrained by
\begin{align}\label{IC:longCT}
    \hat{c}(\nu_t\rvert\lambda_t)=&c^\ast_\#(\eta_t^\infty(\nu_t\rvert\lambda_t))& 
    c^\ast(\nu)\in& C^\ast(\nu)&
    \nu_{t-1}=&\int \nu_{t} \,d\eta_t(\nu_{t}\rvert \nu_{t-1},\lambda_t),
\end{align}
for all $\lambda_t,\nu_{t-1},t$. The fundamental distinction from mediated cheap talk is that the posterior beliefs at each stage are induced by cheap talk. Thus for every $t,\nu_t,\lambda$, the set of posteriors $\supp{\eta_t(\cdot\rvert\nu_t,\lambda)}$ must obtain the same utility, corresponding to eq. \ref{eq:Abuse.IC}. (See \cite{AH03} for a characterization of equilibria.)

With state-independent preferences, in terms of equilibrium mappings between states and receiver choices (using CT=cheap talk):
$$
\text{CT}\subseteq \text{Long CT}\subseteq \text{Mediated CT}\subseteq \text{CT with Receiver Commitment}
$$
where the first three subsets hold for state dependent preferences more generally. While \cite{LR20} show the first subset is an equality in terms of equilibrium sender payoffs in the state-independent preference case, that may fail to hold in nearby (purifiable) models.
\end{document}